\documentclass[conference]{IEEEtran}
\IEEEoverridecommandlockouts

\usepackage{amsmath,amssymb,amsfonts}
\usepackage{amsthm}
\usepackage{bm}
\usepackage{cite}
\usepackage{graphicx}
\usepackage{mathtools}
\usepackage{url}
\usepackage{pgfplots}
\usepackage[hidelinks]{hyperref}
\pgfplotsset{compat=1.17}
\definecolor{figblue}{RGB}{37,99,235}
\definecolor{figpurple}{RGB}{124,58,237}
\definecolor{figred}{RGB}{225,29,72}

\makeatletter
\def\bstctlcite#1{\@bsphack
  \@for\@citeb:=#1\do{%
    \edef\@citeb{\expandafter\@firstofone\@citeb}%
    \if@filesw\immediate\write\@auxout{\string\citation{\@citeb}}\fi}%
  \@esphack}
\makeatother

\newtheorem{theorem}{Theorem}
\newtheorem{proposition}{Proposition}
\newtheorem{corollary}{Corollary}
\newtheorem{lemma}{Lemma}
\newtheorem{remark}{Remark}

\title{\LARGE \bf
Trajectory-Induced Self-Calibration for Hidden-Target Localization Through an Unknown-Pose Range-Bearing Relay}

\author{Yash Bagla\\
Michigan State University\\
{\tt\small yashbagla321@gmail.com}%
\thanks{\raggedright Code and data: \protect\url{https://github.com/yashbagla321/trajectory-induced-self-calibration}, archived at \protect\url{https://doi.org/10.5281/zenodo.21872277}.}}

\begin{document}
\bstctlcite{IEEEexample:BSTcontrol}
\maketitle
\thispagestyle{empty}
\pagestyle{empty}
\begin{abstract}
This paper studies hidden-target localization from range-bearing packets reported by a relay beacon whose global position and yaw are unknown.
The vehicle knows its own trajectory but never directly senses the target; the relay packet contains only local-frame range and bearing to the vehicle and to the hidden target.
Unlike bearing-only network localization, relative-frame localization, and target-enclosing control, the target is neither directly observed in the vehicle frame nor treated as a node in a relative-sensing graph.
The main result characterizes the minimal motion that removes the resulting calibration ambiguity: one vehicle pose leaves a continuous yaw/translation/target gauge, whereas two distinct vehicle-relative observations from one unknown-pose relay constructively determine relay yaw (modulo $2\pi$), relay position, and the anchored target in the noiseless case.
A local rank corollary, a shared-target multi-beacon extension, and a trajectory-spread conditioning lemma connect relay self-calibration to finite-window excitation and native range-bearing estimation.
In Monte Carlo evaluation the estimator recovers the hidden target with $5.5$ mm RMSE, five times below the $30$ mm per-packet range noise and thirteen times more accurate than a naive EKF baseline; it converges to the same accuracy from $2$ m target offsets and $2.4$ rad yaw errors, and Huber weighting preserves millimeter accuracy under $10\%$ outlier corruption that drops the unprotected estimator to a $0.10$ success rate.
Trajectory spread predicts estimator quality: the two weakly excited trajectories carry condition numbers above $100$ with success rates of $0.82$ and $0.70$, while every well-excited trajectory attains full success.
\end{abstract}

\section{Introduction}

Adaptive source localization can drive an autonomous vehicle toward an unknown source using geometric measurements, but range-only formulations generally require multiple stationary beacons for a two-dimensional source~\cite{guler2017adaptive}.
This paper asks a different question: can a vehicle localize a hidden target when the only cooperating sensor is a range-bearing relay whose own global pose is not calibrated?
The difficulty has three parts.
First, the target is hidden from the vehicle.
Second, the relay reports the target only in its unknown local frame.
Third, using that relay at a single vehicle pose leaves a continuous yaw/translation/target gauge, so the target packet cannot be anchored globally until the trajectory supplies another distinct view.

Classical cooperative localization and SLAM typically assume that relevant landmarks are directly sensed, globally anchored, or estimated inside a shared map~\cite{roumeliotis2002distributed,durrantwhyte2006slam,cadena2016slam}.
Gauge freedoms and inconsistent linearization are well known in SLAM and network localization~\cite{huang2010observability,barooah2007estimation,eren2004rigidity}.
Planar absolute-orientation methods also recover a rigid transform from corresponding points~\cite{arun1987least,horn1987closed}.
Closer literature addresses bearing-only network localization with anchors and noise bounds~\cite{shames2013bearing}, relative-frame orientation localization in sensing networks~\cite{piovan2013frame}, and distributed target localization with bearing-based enclosing control~\cite{dou2020target}.
The enclosing-control setting of~\cite{dou2020target} explicitly assumes that all agents share a common coordinate frame and that each agent bearing-senses the target directly; here the absence of a common frame is the problem itself, and the target is never sensed by the vehicle in any frame.
Mobile-beacon localization in sensor networks uses a moving GPS-equipped anchor whose broadcasts let static nodes localize themselves~\cite{han2016mobileanchor}, but that line assumes a common reference frame and direct anchor-to-node measurements and designs trajectories for coverage; here the trajectory is used for identifiability of an unknown relay frame through which a hidden target is reported.
Recent work jointly localizes a target and self-calibrates sensor positions and synchronization offsets from sequential range-angle measurements, and shows that sensor-position self-calibration is impossible in that setting with a single sensor: at least two sensors are required~\cite{jia2025target}.
In the present formulation the known ego trajectory and paired relay-local vehicle/target packets supply the geometric diversity that a second sensor would otherwise provide, so a single relay with unknown translation \emph{and yaw} becomes sufficient after motion-induced excitation.
The present problem differs because the vehicle receives no target coordinates in any globally interpretable frame.
The target is also not a graph node with reciprocal relative measurements as in frame localization, and it is not directly bearing-sensed by each mobile agent as in enclosing-control work.
The closed-form two-view rotation recovery below is a mechanism shared with rigid registration, hand--eye calibration~\cite{tsai1989new}, and relative-frame calibration; the novelty is the relay geometry showing that one unknown-pose beacon becomes sufficient once the vehicle trajectory induces two distinct local vehicle observations, with the hidden target anchored through the same frame rather than observed by the moving sensor.
In calibrated global-frame baselines, beacon variables can often be eliminated after calibration; here the unknown beacon frame itself must be identified from motion before the target estimate is globally anchored.
This distinction appears in cooperative missions where a vehicle can execute a safety-constrained trajectory but cannot directly observe the target.
Receding-horizon chance-constrained planning and impact-aware autonomous operation motivate this separation between safe ego motion and uncertain external target information~\cite{bagla2019receding,johnson2024impact}.
Seen from this planning perspective, the stored trajectory is both the vehicle path and the calibration experiment that makes a teammate's local report usable; the relevant question is whether motion has made the beacon frame identifiable before target seeking relies on it.

The contributions are:
\begin{enumerate}
\item a two-view constructive identifiability theorem for an unknown-pose local-frame relay model of hidden-target localization, showing that two distinct vehicle-relative observations suffice and recovering relay yaw, relay position, and the hidden target in closed form, together with a local rank corollary and a shared-target multi-beacon extension;
\item a one-pose gauge theorem establishing that this excitation is minimal, since a single vehicle pose leaves a continuous yaw/translation/target ambiguity and static use of the relay is insufficient;
\item a trajectory-spread conditioning lemma showing that the centered spread $S_v$ is simultaneously a finite-window excitation margin and a yaw/translation Schur complement, together with a local weighted least-squares consequence for the noisy estimator;
\item reproducible simulations that test conditioning, EKF initialization, robust losses, target-packet averaging, dropouts, outliers, and vehicle-pose errors.
\end{enumerate}
The claim is local in the presence of noise: the theory identifies the geometry that removes the gauge and explains conditioning, while constructive initialization and multistart refinement address remote nonlinear minima numerically.

\begin{figure}[t]
\centering
\includegraphics[width=0.98\linewidth]{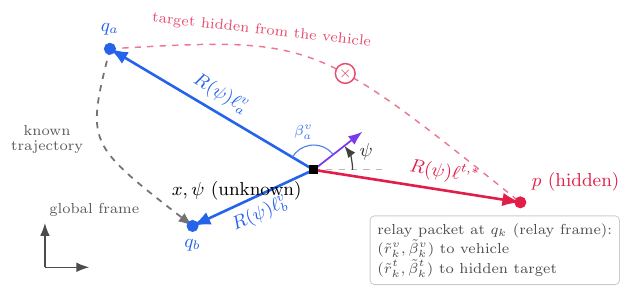}
\caption{Hidden-target self-calibration geometry. The vehicle traverses known poses $q_a,q_b$; the relay's position $x$ and heading $\psi$ (purple arrow, measured against the global-parallel dashed axis) are unknown; and the vehicle never senses the target $p$ (crossed dashed link). Each relay packet reports local-frame range-bearing pairs to the vehicle and to the target, with bearings measured from the relay heading ($\beta_a^v$ shown), and the solid arrows are the global realizations of the reported vectors under $R(\psi)$. Two distinct vehicle vectors determine $x,\psi$, and $\ell^{t,*}$ then anchors $p$.}
\label{fig:acc_geometry}
\end{figure}

\section{Measurement Model}

Let $p\in\mathbb{R}^2$ be a fixed hidden target and let beacon $i$ (the relay of the title; the two terms are used interchangeably) have unknown position $x_i\in\mathbb{R}^2$ and yaw $\psi_i$.
The beacon-to-global rotation is $R_i=R(\psi_i)$.
At vehicle pose $q_k$, beacon $i$ reports noisy local-frame packets
\[
\tilde\ell_{ik}^{v}=\tilde r_{ik}^{v}
\begin{bmatrix}\cos\tilde\beta_{ik}^{v}\\ \sin\tilde\beta_{ik}^{v}\end{bmatrix},\qquad
\tilde\ell_{ik}^{t}=\tilde r_{ik}^{t}
\begin{bmatrix}\cos\tilde\beta_{ik}^{t}\\ \sin\tilde\beta_{ik}^{t}\end{bmatrix}.
\]
The corresponding noiseless vectors are $\ell_{ik}^{v}$ and $\ell_i^{t,*}$, with
\begin{equation}
q_k=x_i+R_i\ell_{ik}^{v},\qquad
p=x_i+R_i\ell_i^{t,*}. \label{eq:model}
\end{equation}
Fig.~\ref{fig:acc_geometry} illustrates this geometry for a single beacon, together with the two-view motion that Theorem~\ref{thm:constructive} shows removes the calibration gauge.
The state is $z=[p^T,x_1^T,\psi_1,\ldots,x_N^T,\psi_N]^T$.
In the nominal Monte Carlo model the target packet is remeasured at each stored vehicle pose with independent noise; a separate validation row uses only one target packet to expose the averaging gain from repeated target measurements.
Packets are synchronized to stored $q_k$, beacon identity is known, and bearings/yaws are wrapped to $(-\pi,\pi]$.
The vehicle pose is treated as known by the estimator, because it is the global anchor that breaks the beacon-frame gauge; Section~\ref{sec:validation} separately corrupts $q_k$ with i.i.d.\ noise and random-walk drift to quantify this limitation.
The target is static during a batch window, so repeated target packets improve statistics but do not add new target geometry.
All ranges used in the rank arguments are assumed nonzero, so the polar map from local range-bearing packets to local vectors is a local diffeomorphism.
The superscripts $v$ and $t$ identify packet endpoints, not globally oriented bearings: both directions are local to the unknown beacon yaw, which is the source of the gauge below.

The estimator minimizes native range-bearing residuals, not Cartesianized residuals.
For one packet, with $\angle v:=\operatorname{atan2}(v_2,v_1)$ for nonzero $v\in\mathbb{R}^2$,
\[
g_{ik}(z)=
\begin{bmatrix}
\|q_k-x_i\|\\
\angle(q_k-x_i)-\psi_i\\
\|p-x_i\|\\
\angle(p-x_i)-\psi_i
\end{bmatrix},\quad
y_{ik}=
\begin{bmatrix}
\tilde r_{ik}^{v}\\ \tilde\beta_{ik}^{v}\\
\tilde r_{ik}^{t}\\ \tilde\beta_{ik}^{t}
\end{bmatrix}.
\]
With wrapped angle differences, $r_{ik}(z)=y_{ik}-g_{ik}(z)$ is weighted by
$W=\mathrm{diag}(\sigma_r^{-2},\sigma_\theta^{-2},\sigma_r^{-2},\sigma_\theta^{-2})$.
The batch objective is
\begin{equation}
J(z)=\frac{1}{2}\sum_{i=1}^{N}\sum_{k=1}^{K} r_{ik}(z)^TWr_{ik}(z). \label{eq:cost}
\end{equation}
Using native range-bearing residuals avoids mixing meter and radian errors in an unweighted Cartesian objective.
It also makes the information matrix and the EKF baseline use the same measurement units.

\section{Gauge Barrier and Motion-Induced Recovery}

For a single beacon we drop the beacon index and write $x,\psi,\ell_k^v,\ell^{t,*}$ and $R=R(\psi)$; when only one pose is involved we also drop the pose index and write $q,\ell^v$.

\begin{theorem}[One-pose gauge barrier]
\label{thm:gauge}
With one beacon and one vehicle pose, $(p,x,\psi)$ is not identifiable.
\end{theorem}

\begin{proof}
For measured local vectors $\ell^v,\ell^{t,*}$ and any $\alpha\in\mathbb{R}$, set
$\psi(\alpha)=\psi+\alpha$, $x(\alpha)=q-R(\psi+\alpha)\ell^v$, and
$p(\alpha)=x(\alpha)+R(\psi+\alpha)\ell^{t,*}$.
Then the same measurements are reproduced for all $\alpha$, giving a one-dimensional yaw/translation/target gauge.
The tangent direction at $\alpha=0$ is
\begin{align*}
\delta\psi&=1,\qquad \delta x=-R(\psi)S\ell^v,\\
\delta p&=R(\psi)S(\ell^{t,*}-\ell^v),
\end{align*}
where $S=\left[\begin{smallmatrix}0&-1\\1&0\end{smallmatrix}\right]$.
Thus a rank-based local test must have at least this null direction when only one vehicle-relative observation is available.
\end{proof}

\begin{remark}[Relation to network localizability]
Viewed as a static relative-sensing network, the one-beacon configuration is nonlocalizable in the frameworks of rigidity, bearing-only localization, and relative-frame localization~\cite{eren2004rigidity,shames2013bearing,piovan2013frame}.
The beacon is an unanchored node with unknown orientation, the target is a degree-one node sensed only by that beacon, and no reciprocal edge exists because the vehicle does not directly sense the beacon.
Orientation recovery in~\cite{piovan2013frame} requires reciprocal angle-of-arrival information, while bearing localizability in~\cite{shames2013bearing} uses bearings in a common frame with anchors.
Theorem~\ref{thm:constructive} shows that ego motion substitutes for these missing resources: the directed beacon-to-vehicle edge, evaluated at two distinct known vehicle positions, becomes a time-indexed pair of measurements to anchors; once the relay frame is identified, the degree-one target is anchored through it.
\end{remark}

\begin{theorem}[Motion-induced two-view recovery]
\label{thm:constructive}
Consider one unknown-pose beacon and noiseless measurements at two vehicle poses $q_a,q_b$.
If $\ell_a^v\neq\ell_b^v$, then $p$ and $x$ are uniquely determined, and $\psi$ is determined modulo $2\pi$.
\end{theorem}

\begin{proof}
Subtracting the two vehicle equations in \eqref{eq:model} gives
\[
q_b-q_a=R(\psi)(\ell_b^v-\ell_a^v).
\]
Since $\ell_b^v-\ell_a^v\neq0$,
\[
\psi=\angle(q_b-q_a)-\angle(\ell_b^v-\ell_a^v)\pmod{2\pi}.
\]
Then $x=q_a-R(\psi)\ell_a^v$ and $p=x+R(\psi)\ell^{t,*}$ by substitution.
No other yaw modulo $2\pi$ can rotate the nonzero vector $\ell_b^v-\ell_a^v$ into $q_b-q_a$, so the recovery is unique.
If $\ell_a^v=\ell_b^v$, the yaw equation collapses and the one-pose gauge in Theorem~\ref{thm:gauge} remains.
\end{proof}

\begin{corollary}[Local rank after gauge removal]
\label{cor:rank}
For one beacon and $K$ stored poses, the stacked residual Jacobian of \eqref{eq:cost} has rank five at the true state if and only if $\{\ell_k^v\}_{k=1}^{K}$ contains at least two distinct vectors.
\end{corollary}

\begin{proof}
Because the polar packet map is locally invertible for nonzero ranges, it is enough to study the Cartesian local-vector constraints in \eqref{eq:model}.
For one beacon and perturbation $\delta z=[\delta p^T,\delta x^T,\delta\psi]^T$, the linearized constraint matrix is
\[
\mathcal O_K=
\begin{bmatrix}
0&I&R S\ell_1^v\\
I&-I&-R S\ell^{t,*}\\
\vdots&\vdots&\vdots\\
0&I&R S\ell_K^v\\
I&-I&-R S\ell^{t,*}
\end{bmatrix}.
\]
The native range-bearing Jacobian is obtained by left multiplying this local-vector matrix by the block-diagonal polar Jacobian of each nonzero vehicle and target vector.
Hence the range-bearing residual Jacobian and $\mathcal O_K$ have the same column rank on the nonzero-range set used by the estimator.
The nullspace equations are
\[
0=\delta x+R S\ell_k^v\delta\psi,\qquad
0=\delta p-\delta x-RS\ell^{t,*}\delta\psi .
\]
If all $\ell_k^v$ are identical, the nullspace equations force every solution onto the gauge direction of Theorem~\ref{thm:gauge}, so the nullspace is exactly one-dimensional and the rank is four.
If two vectors differ, subtracting the corresponding vehicle equations gives $RS(\ell_a^v-\ell_b^v)\delta\psi=0$; since $R$ and $S$ are nonsingular and $\ell_a^v\neq\ell_b^v$, $\delta\psi=0$, hence $\delta x=0$ and $\delta p=0$.
Thus the one-beacon rank is five.
\end{proof}

Therefore, for the one-beacon unknown-pose model, one vehicle-relative observation is insufficient, while two distinct vehicle-relative observations are sufficient for local identification of the five-dimensional state.
One beacon is locally sufficient only as a trajectory-excited sensor, not as a static calibrated landmark.

\begin{corollary}[Shared-target beacon augmentation]
\label{cor:augmentation}
Suppose one beacon satisfies Corollary~\ref{cor:rank}, so the target $p$ is locally identified.
For any additional beacon $j$, one synchronized packet at vehicle pose $q_k$ locally identifies $(x_j,\psi_j)$ if $\ell_j^{t,*}\neq\ell_{jk}^{v}$.
In particular,
\[
\psi_j=\angle(p-q_k)-\angle(\ell_j^{t,*}-\ell_{jk}^{v}),\qquad
x_j=q_k-R(\psi_j)\ell_{jk}^{v}.
\]
Thus the $N$-beacon Jacobian has rank $2+3N$ under one trajectory-excited anchor beacon and one noncoincident target/vehicle packet for each remaining beacon.
\end{corollary}

\begin{proof}
Once $p$ is locally fixed by an excited beacon, subtracting $q_k=x_j+R_j\ell_{jk}^{v}$ from $p=x_j+R_j\ell_j^{t,*}$ gives
\[
p-q_k=R_j(\ell_j^{t,*}-\ell_{jk}^{v}).
\]
If the local difference is nonzero, the same orientation argument as Theorem~\ref{thm:constructive} fixes $\psi_j$ modulo $2\pi$, and then $x_j$ follows by substitution.
At the differential level, ordering the variables as $(p,x_1,\psi_1,x_2,\psi_2,\ldots)$ and grouping rows by beacon gives a block lower-triangular Jacobian after the anchor beacon fixes $p$; each noncoincident additional packet contributes a full-column-rank three-column beacon-pose block.
Repeated rows for the same static target packet can improve noise averaging but add no new rank.
Applying the argument to each additional beacon gives the stated rank.
\end{proof}

\begin{remark}[Auxiliary-beacon margin]
For beacon $j$, the separation $\|\ell_j^{t,*}-\ell_{jk}^{v}\|$ plays the same role as the two-view separation $\|\ell_b^v-\ell_a^v\|$ for the anchor beacon.
Coincident target and vehicle directions in the beacon frame are the degenerate geometry, and the closed-form seed error for $(x_j,\psi_j)$ scales inversely with this separation.
\end{remark}

\emph{Finite-window observability interpretation.}
For one beacon, the finite-window local-vector output
\[
\Phi_K(z)=\{R(\psi)^T(q_k-x),\,R(\psi)^T(p-x)\}_{k=1}^{K}
\]
has Jacobian rank four if all $\ell_k^v$ are identical and rank five otherwise.
The only possible loss of rank is the yaw/translation gauge in Theorem~\ref{thm:gauge}.
The differential of the vehicle component is
\[
\delta \ell_k^v=-R^T\delta x-S\ell_k^v\delta\psi ,
\]
and the target component gives
\[
\delta \ell^t=R^T(\delta p-\delta x)-S\ell^{t,*}\delta\psi .
\]
Setting these differentials to zero gives the same nullspace equations used in Corollary~\ref{cor:rank}.
Thus the local output codistribution has dimension four in the one-pose gauge case and five once motion creates a second distinct local observation.
The known trajectory fixes the global translation/rotation gauge that would appear in a purely relative SLAM problem; the remaining unobservable direction is not global-frame gauge but a beacon-frame self-calibration gauge.
Vehicle motion is therefore not merely additional data; it is the input that changes the measurement codistribution.

\begin{lemma}[Finite-window excitation margin]
\label{lem:spread}
For one beacon, define the centered local vehicle-vector spread
\[
S_v=\sum_{k=1}^{K}\|\ell_k^v-\bar\ell^v\|^2,\qquad
\bar\ell^v=\frac{1}{K}\sum_{k=1}^{K}\ell_k^v .
\]
For any fixed yaw perturbation $\delta\psi$,
\[
\min_{\delta x}\sum_k\|\delta x+R(\psi)S\ell_k^v\delta\psi\|^2
=\delta\psi^2S_v .
\]
The yaw direction is unseparated from beacon translation exactly when $S_v=0$, and its separation from translation is set by the centered spread of the local vehicle vectors.
Moreover, in the noiseless model $S_v=\sum_k\|q_k-\bar q\|^2$, so finite-window vehicle motion directly controls the self-calibration margin.
\end{lemma}

\begin{proof}
Linearizing the vehicle part of \eqref{eq:model} gives
$0=\delta x+R(\psi)S\ell_k^v\delta\psi$.
Eliminating the best translation perturbation gives the displayed identity by centering the vectors $\ell_k^v$.
Finally, $\ell_k^v=R^T(q_k-x)$ implies $\ell_k^v-\bar\ell^v=R^T(q_k-\bar q)$, and rotation preserves norms.
\end{proof}

The metric $S_v$ is independent of global translation and depends only on how much the vehicle trajectory changes its appearance in the beacon frame.
It is therefore a direct design variable for motion planning: a small trajectory baseline gives tightly clustered local vectors and weak yaw/translation separation even at full rank, whereas an excited trajectory gives the yaw Jacobian column a component beacon translation cannot reproduce~\cite{leny2018localizability}.
Equivalently, $S_v=K^{-1}\sum_{a<b}\|\ell_a^v-\ell_b^v\|^2$; hence $S_v>0$ if and only if at least one pair of vehicle views differs, while its magnitude aggregates the separation of all informative viewpoints.
In local-vector coordinates with isotropic vehicle-vector weight $w_vI$, the yaw Schur complement after eliminating the target and beacon translation is exactly $w_vS_v$.

\begin{corollary}[Spread-certified pair selection]
\label{cor:pairselection}
Let $d_\star=\max_{a<b}\|\ell_b^v-\ell_a^v\|$ over a window of $K\ge2$ views, and suppose every measured vehicle vector obeys $\|\tilde\ell_k^v-\ell_k^v\|\le\epsilon$.
Select the pair with largest measured separation and form the two-view yaw estimate of Theorem~\ref{thm:constructive}.
If $d_\star>8\epsilon$, its true separation $d_s$ and yaw error satisfy
\begin{equation}
d_s\ge d_\star-4\epsilon>4\epsilon,\qquad
|\hat\psi-\psi|\le\frac{4\pi\epsilon}{d_\star}
\le4\pi\epsilon\sqrt{\frac{K-1}{2S_v}}. \label{eq:pairbound}
\end{equation}
Thus the full-window spread certifies a finite-noise two-view initializer even though the constructive recovery itself uses only one selected pair.
\end{corollary}

\begin{proof}
For the true maximizer, measurement perturbations reduce its measured separation by at most $2\epsilon$.
The selected pair has at least that measured separation, and converting it back to true separation loses at most another $2\epsilon$, giving $d_s\ge d_\star-4\epsilon$.
The angle between a nonzero difference vector and a perturbation of norm at most $2\epsilon$ is at most $2\pi\epsilon/d_s$ when $d_s>4\epsilon$, hence the first yaw bound after using $d_s>d_\star/2$.
Finally, the pairwise identity for $S_v$ implies
$d_\star^2\ge 2S_v/(K-1)$, which gives the second bound.
\end{proof}

Corollary~\ref{cor:pairselection} explains why the initializer searches the entire stored window rather than accepting the first two packets.
Additional views need not add state dimension or change the noiseless sufficiency result; they improve the chance of finding a well-separated pair and therefore enlarge the certified initialization basin before nonlinear refinement begins.

\begin{proposition}[Native range-bearing information bound]
\label{prop:polarinfo}
Let $r_k^v=\|\ell_k^v\|>0$, $u_k=\ell_k^v/r_k^v$, $u_k^\perp=Su_k$, and let the native vehicle-packet covariance be $\Sigma=\operatorname{diag}(\sigma_r^2,\sigma_\theta^2)$.
The information induced on a local-vector perturbation is
\[
Q_k=\frac{u_ku_k^T}{\sigma_r^2}
 +\frac{u_k^\perp(u_k^\perp)^T}{(r_k^v)^2\sigma_\theta^2}.
\]
Define $m=\min_k\lambda_{\min}(Q_k)>0$ and $M=\max_k\lambda_{\max}(Q_k)<\infty$.
After eliminating beacon translation, the scalar yaw information $\mathcal I_\psi$ satisfies
\begin{equation}
mS_v\le \mathcal I_\psi\le MS_v. \label{eq:polarbound}
\end{equation}
Consequently $S_v=0$ is exactly the native-residual gauge case, while for $S_v>0$ the linearized yaw covariance obeys
$1/(MS_v)\le[ F^{-1}]_{\psi\psi}\le1/(mS_v)$ after marginalizing the unconstrained target block.
\end{proposition}

\begin{proof}
For the polar map $\pi(\ell)=(\|\ell\|,\operatorname{atan2}(\ell_y,\ell_x))$, its Jacobian is $D\pi(\ell_k^v)=[u_k^T;\ (u_k^\perp)^T/r_k^v]$, hence $Q_k=D\pi^T\Sigma^{-1}D\pi$ and its radial and tangential eigenvalues are $\sigma_r^{-2}$ and $((r_k^v)^2\sigma_\theta^2)^{-1}$.
Eliminating beacon translation from the vehicle rows gives
\[
\mathcal I_\psi=\min_a\sum_k(a+S\ell_k^v)^TQ_k(a+S\ell_k^v).
\]
Bounding every $Q_k$ between $mI$ and $MI$, then using Lemma~\ref{lem:spread}, gives the lower bound; choosing the unweighted minimizing translation $a=-S\bar\ell^v$ gives the upper bound.
Target rows contribute no additional calibration information after $p$ is marginalized because their common local vector can be absorbed by a target perturbation.
Inverting the scalar Schur complement yields the covariance bounds.
\end{proof}

Proposition~\ref{prop:polarinfo} makes the range dependence explicit rather than hiding it in a continuity statement.
On an operating set $r_k^v\in[r_{\min},r_{\max}]$, one may use
$m=\min\{\sigma_r^{-2},(r_{\max}^2\sigma_\theta^2)^{-1}\}$ and
$M=\max\{\sigma_r^{-2},(r_{\min}^2\sigma_\theta^2)^{-1}\}$.
Thus $S_v$ controls native-residual information up to known radial/tangential noise anisotropy, while range determines the conversion from angular noise to transverse position uncertainty.

\begin{corollary}[Finite-window excitation]
\label{cor:excitation}
In the noiseless one-beacon model, any batch window containing at least two distinct known vehicle positions has $S_v>0$ and removes the one-pose gauge.
Consequently, a controller can enforce a checkable excitation condition by retriggering motion until $\sum_k\|q_k-\bar q\|^2$ exceeds a chosen threshold.
\end{corollary}

\begin{proof}
Lemma~\ref{lem:spread} gives $S_v=\sum_k\|q_k-\bar q\|^2$.
This quantity is zero if and only if all stored vehicle positions are identical.
If two positions differ, $S_v>0$, hence the local vehicle vectors cannot all be identical and Corollary~\ref{cor:rank} gives rank five.
\end{proof}

\begin{corollary}[Local weighted least-squares recovery]
\label{cor:gn}
Let $y=g(z^\star)$ be noiseless data satisfying Corollary~\ref{cor:rank}, and let $F=J_g(z^\star)^TWJ_g(z^\star)$.
Then $F\succ0$, $z^\star$ is a locally isolated minimizer of \eqref{eq:cost}, and damped Gauss--Newton initialized sufficiently near $z^\star$ converges locally to $z^\star$~\cite{nocedal2006numerical}.
For noisy measurements $y=g(z^\star)+\eta$ satisfying the same rank condition, the local weighted least-squares estimator obeys
\[
\hat z-z^\star=(J_g^TWJ_g)^{-1}J_g^TW\eta+O(\|\eta\|^2),
\]
with all Jacobians evaluated at $z^\star$.
If $W$ is the inverse packet-noise covariance, then $\mathrm{cov}(\hat z)\approx F^{-1}$ in the usual small-noise approximation~\cite{barshalom2001estimation}.
\end{corollary}

\begin{proof}
Corollary~\ref{cor:rank} makes $J_g(z^\star)$ full column rank, so it contains a nonsingular square minor; applying the inverse function theorem to the corresponding output submap gives local injectivity of $g$ near $z^\star$, so $z^\star$ is a locally isolated zero-residual state and $F$ is nonsingular.
The first-order optimality condition for \eqref{eq:cost} is $J_g(\hat z)^TW(y-g(\hat z))=0$.
Expanding this equation at $(z^\star,\eta)=(z^\star,0)$ and retaining first-order terms gives
$F(\hat z-z^\star)=J_g^TW\eta$.
Because $F\succ0$, the expression follows by inversion; the covariance statement follows by substituting $\mathrm{cov}(\eta)=W^{-1}$.
\end{proof}

The constructive proof also provides the initializer used in the implementation: choose two distinct vehicle-relative observations, compute $\psi$, back-substitute $x$, average the transformed target packets for $p$, and refine with weighted Gauss--Newton.
Corollary~\ref{cor:gn} is local: it does not claim global performance, but it predicts the sensitivity trend seen in the trajectory sweep and explains why small $S_v$ produces poor numerical conditioning even after rank is achieved.
The perturbation argument follows the noisy-localizability methodology of~\cite{shames2013bearing,anderson2010formal}, transported to the unknown-frame relay model; the novelty is the motion-induced rank condition for relay self-calibration.

\section{Estimator and Target Seeking}

The implementation uses analytic Jacobians for the range-bearing residuals in \eqref{eq:cost}, Levenberg damping initialized at $10^{-3}$, and at most 80 iterations.
The robust baseline applies a Huber loss to the whitened residuals.
The two-view EKF uses the constructive seed after two informative views, giving a fair comparison against the batch estimator; the naive EKF instead starts from the uncalibrated prior and is reported separately as a gauge-sensitive initialization stress test.
For batch refinement, the first multistart candidate is always the constructive or nominal seed, and remaining candidates are perturbations used only to enlarge the attraction basin.
Candidates are ranked by the whitened cost in \eqref{eq:cost}, so multistart is not a separate estimator but a numerical safeguard against remote local minima after the local gauge has been removed.

The vehicle obeys $\dot q=u$ and uses
\begin{align}
u&=-k(q-\hat p)+u^{\rm exp}(t), \nonumber\\
u^{\rm exp}(t)&=Ae^{-\lambda(t-t_0)}
\begin{bmatrix}\cos\omega t\\ \sin\omega t\end{bmatrix}. \label{eq:controller}
\end{align}
The excitation epoch $t_0$ can be reset when $S_v$ or $\sigma_{\min}(F)$ remains small, or when the batch cost stalls.
Circular excitation is a simple way to satisfy Corollary~\ref{cor:excitation} before target seeking dominates.
The spread $S_v$ is therefore also a natural scalar objective for information-seeking motion, although this paper does not claim a new optimal excitation law.
If only finitely many resets occur, $\hat p(t)\to p$, and the excitation decays integrably after the final reset, then $q(t)\to p$ by input-to-state stability of $\dot q=-k(q-p)+k(\hat p-p)+u^{\rm exp}$~\cite{khalil2002nonlinear}.
The controller is deliberately simple: the contribution is the calibration geometry that tells the controller when motion is informative, not a new optimal-control law.
Two subsequent studies develop complementary control consequences of this calibration geometry.
The closed-loop analysis in~\cite{bagla2027excitationsupervised} formalizes the informal reset rule into an explicit excitation-supervision algorithm with a finite-acquisition guarantee and an accuracy-driven excitation budget while retaining the globally known vehicle trajectory assumed here.
The anchor-free formulation in~\cite{bagla2026anchorfree} instead removes that trajectory anchor: it retains the unavoidable global $SE(2)$ gauge, recovers only the task-relevant quotient from integrated odometry, and propagates the resulting cross-view pose correlations into an online uncertainty certificate.

\section{Validation}
\label{sec:validation}

All simulations use the same C++ estimator core.
Unless stated otherwise, Monte Carlo entries use 100 trials for the main summary, 50 trials for expanded stress rows, $K=80$ stored vehicle poses, $\sigma_r=0.03$ m, $\sigma_\theta=0.006$ rad, and 95\% confidence intervals.
Success means target error below $0.05$ m, beacon-position error below $0.05$ m, and yaw error below $0.05$ rad for full-state methods.
The ROS~2/Gazebo package links to the same core for software-in-the-loop reproduction of logged trials.
The validation is organized as theorem stress-testing: Table~\ref{tab:excitation} verifies the gauge/rank transition, Figure~\ref{fig:conditioning} tests $S_v$ as a conditioning predictor, Table~\ref{tab:methods} reports estimator consequences with fair two-view initialization, and Table~\ref{tab:stress} stresses the assumptions under weak motion, outliers, dropouts, and vehicle-pose error.
The complete validation reports the per-trajectory conditioning sweep (Table~\ref{tab:trajectories}), the Monte Carlo model comparison (Table~\ref{tab:mc}) with its per-trial error distributions (Fig.~\ref{fig:ecdf}), the full noise grid (Fig.~\ref{fig:noise}), and the initialization, outlier, degradation, and beacon-geometry sweeps (Tables~\ref{tab:poorinit} and~\ref{tab:outliers}--\ref{tab:geometry}).

\subsection{Excitation and Conditioning}

\begin{table}[t]
\centering
\caption{One-beacon excitation study, noise-free.}
\label{tab:excitation}
\scriptsize
\begin{tabular}{l c c c c c}
\hline
Case & Poses & Rank & $S_v$ & $\sigma_{\min}$ & Target (m)\\
\hline
Single pose & 1 & 4 & 0.00 & 0.00 & 1.8374\\
Two poses & 2 & 5 & 9.04 & 22.18 & 0.0000\\
Full trajectory & 80 & 5 & 314.21 & 219.12 & 0.0000\\
\hline
\end{tabular}
\end{table}

Table~\ref{tab:excitation} verifies Theorems~\ref{thm:gauge}--\ref{thm:constructive}: a single pose admits a zero-cost wrong target, while two distinct poses recover the state exactly.
In the trajectory sweep of Table~\ref{tab:trajectories}, stationary motion has $S_v=0$ and rank four; short-line and repeated-view trajectories reach rank five but have whitened-Jacobian condition numbers $\kappa=117.0$ and $130.3$; excited and circular trajectories have $\kappa$ near $10$.
Thus observability is binary, but estimator quality depends on the spread margin after the gauge is removed; Figure~\ref{fig:conditioning} plots this trend directly.
Condition numbers use one fixed state parameterization (meters, radians) and are compared across trajectories, not treated as coordinate-invariant.
The short-line and repeated-view cases also show why the constructive theorem should not be interpreted as saying that any nonzero motion is equally useful.

\begin{table}[t]
\centering
\caption{Trajectory conditioning sweep.}
\label{tab:trajectories}
\scriptsize
\setlength{\tabcolsep}{3.5pt}
\begin{tabular}{c l c c c c c}
\hline
\# & Trajectory & $S_v$ & $\sigma_{\min}$ & $\kappa$ & Target (m) & Success\\
\hline
-- & stationary & 0.0 & 0.00 & -- & 6.068 & 0.00\\
1 & short line & 3.6 & 18.94 & 117.0 & 0.0252 & 0.82\\
2 & repeated views & 4.3 & 17.01 & 130.3 & 0.0481 & 0.70\\
3 & low-curvature arc & 132.7 & 102.07 & 21.6 & 0.0077 & 1.00\\
4 & collinear pass & 214.4 & 295.03 & 12.2 & 0.0046 & 1.00\\
5 & excited figure-eight & 222.1 & 208.69 & 10.8 & 0.0059 & 1.00\\
6 & figure-eight & 252.8 & 229.24 & 9.8 & 0.0055 & 1.00\\
7 & line & 267.9 & 165.76 & 13.3 & 0.0059 & 1.00\\
8 & excited & 314.2 & 219.12 & 10.4 & 0.0052 & 1.00\\
9 & circle & 352.7 & 252.99 & 8.9 & 0.0048 & 1.00\\
\hline
\end{tabular}\\[3pt]
\parbox{\columnwidth}{\raggedright\scriptsize Note: marker numbers match Fig.~\ref{fig:conditioning}; 50 trials per row. $S_v$, $\sigma_{\min}$, and $\kappa$ are noise-free values at the true state; the stationary row is rank four, so $\kappa$ is unbounded there, and every other row is rank five.}
\end{table}

Table~\ref{tab:trajectories} lists the full sweep behind Figure~\ref{fig:conditioning}.
The boundary is sharp in spread rather than in rank: the two weakly excited trajectories with $S_v<5$ carry condition numbers above $100$ and success rates of $0.82$ and $0.70$, while every trajectory with $S_v>100$ has $\kappa\le21.6$ and full success.
The collinear pass is the instructive row: the vehicle positions are collinear, yet its spread is large, and Lemma~\ref{lem:spread} depends only on the centered spread of the stored poses, not on the shape of the path, so it is among the best-conditioned rows in the sweep.
Degeneracy in this problem is about how little the trajectory moves in the beacon frame, not about the curve it traces.

\begin{figure}[t]
\centering
\begin{tikzpicture}
\begin{axis}[
    width=\columnwidth, height=0.58\columnwidth,
    ylabel={$\log_{10}\kappa$},
    xmin=0.35, xmax=2.75, ymin=0.8, ymax=2.35,
    grid=major, grid style={black!10},
    tick label style={font=\footnotesize},
    label style={font=\footnotesize},
    legend columns=2,
    legend style={at={(0.5,1.03)}, anchor=south, draw=none, fill=none, font=\footnotesize, column sep=8pt},
]
\addplot[only marks, mark=*, mark size=1.8pt, figred] table[x=lsv, y=lkappa, restrict expr to domain={\thisrow{idx}}{1:2}] {conditioning_points.dat};
\addlegendentry{weakly excited ($S_v<5$)}
\addplot[only marks, mark=*, mark size=1.8pt, figblue] table[x=lsv, y=lkappa, restrict expr to domain={\thisrow{idx}}{3:9}] {conditioning_points.dat};
\addlegendentry{well excited ($S_v>100$)}
\node[anchor=north west, font=\footnotesize] at (rel axis cs:0.02,0.97) {(a)};
\node[anchor=north, font=\scriptsize, inner sep=2.5pt] at (axis cs:0.559,2.068) {1};
\node[anchor=south, font=\scriptsize, inner sep=2.5pt] at (axis cs:0.628,2.115) {2};
\node[anchor=south, font=\scriptsize, inner sep=2.5pt] at (axis cs:2.123,1.335) {3};
\node[anchor=east, font=\scriptsize, inner sep=2.5pt] at (axis cs:2.331,1.087) {4};
\node[anchor=north east, font=\scriptsize, inner sep=2pt] at (axis cs:2.347,1.035) {5};
\node[anchor=north, font=\scriptsize, inner sep=2.5pt] at (axis cs:2.403,0.991) {6};
\node[anchor=south, font=\scriptsize, inner sep=2.5pt] at (axis cs:2.428,1.125) {7};
\node[anchor=west, font=\scriptsize, inner sep=2.5pt] at (axis cs:2.497,1.017) {8};
\node[anchor=north west, font=\scriptsize, inner sep=2pt] at (axis cs:2.547,0.949) {9};
\end{axis}
\end{tikzpicture}\\[2pt]
\begin{tikzpicture}
\begin{axis}[
    width=\columnwidth, height=0.58\columnwidth,
    xlabel={$\log_{10}S_v$}, ylabel={RMSE (m)},
    xmin=0.35, xmax=2.75, ymin=0, ymax=0.054,
    ytick={0,0.01,0.02,0.03,0.04,0.05},
    scaled y ticks=false,
    yticklabel style={/pgf/number format/fixed},
    grid=major, grid style={black!10},
    tick label style={font=\footnotesize},
    label style={font=\footnotesize},
]
\addplot[only marks, mark=*, mark size=1.8pt, figred] table[x=lsv, y=rmse, restrict expr to domain={\thisrow{idx}}{1:2}] {conditioning_points.dat};
\addplot[only marks, mark=*, mark size=1.8pt, figblue] table[x=lsv, y=rmse, restrict expr to domain={\thisrow{idx}}{3:9}] {conditioning_points.dat};
\node[anchor=north west, font=\footnotesize] at (rel axis cs:0.02,0.97) {(b)};
\node[anchor=west, font=\scriptsize, inner sep=2.5pt] at (axis cs:0.559,0.02523) {1};
\node[anchor=west, font=\scriptsize, inner sep=2.5pt] at (axis cs:0.628,0.04807) {2};
\node[anchor=south, font=\scriptsize, inner sep=2.5pt] at (axis cs:2.123,0.00772) {3};
\node[anchor=north east, font=\scriptsize, inner sep=2pt] at (axis cs:2.331,0.00459) {4};
\node[anchor=south east, font=\scriptsize, inner sep=2pt] at (axis cs:2.347,0.00594) {5};
\node[anchor=north, font=\scriptsize, inner sep=2.5pt] at (axis cs:2.403,0.00551) {6};
\node[anchor=south, font=\scriptsize, inner sep=2.5pt] at (axis cs:2.428,0.00588) {7};
\node[anchor=south, font=\scriptsize, inner sep=2.5pt] at (axis cs:2.497,0.00517) {8};
\node[anchor=north west, font=\scriptsize, inner sep=2pt] at (axis cs:2.547,0.00482) {9};
\end{axis}
\end{tikzpicture}
\caption{Trajectory spread $S_v$ versus (a) conditioning and (b) target RMSE over 50 trials. Markers 1--9 are the trajectories numbered in Table~\ref{tab:trajectories}; the two weakly excited trajectories (red) sit an order of magnitude above the rest in both panels. Stationary motion ($S_v=0$) is omitted from the log axis.}
\label{fig:conditioning}
\end{figure}
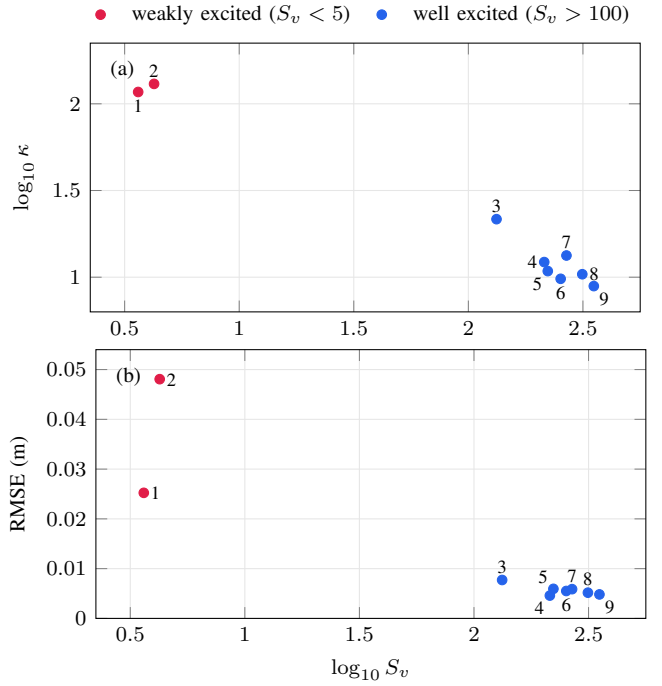

\subsection{Estimator and Model Comparison}

\begin{table}[t]
\centering
\caption{Estimator comparison under Gaussian noise.}
\label{tab:methods}
\scriptsize
\setlength{\tabcolsep}{1.0pt}
\resizebox{\columnwidth}{!}{%
\begin{tabular}{@{}l c c c c c@{}}
\hline
Method/case & Target (m) & Beacon (m) & Yaw (rad) & Time (ms) & Success\\
\hline
Batch GN & $0.00549\!\pm\!0.00077$ & $0.00247\!\pm\!0.00038$ & $0.00072\!\pm\!0.00015$ & 2.03 & 1.00\\
Huber GN & $0.00503\!\pm\!0.00067$ & $0.00203\!\pm\!0.00036$ & $0.00072\!\pm\!0.00015$ & 5.41 & 1.00\\
Multistart GN & $0.00464\!\pm\!0.00053$ & $0.00233\!\pm\!0.00037$ & $0.00072\!\pm\!0.00015$ & 9.10 & 1.00\\
Sliding-window GN & $0.01925\!\pm\!0.00352$ & $0.04657\!\pm\!0.00905$ & $0.01049\!\pm\!0.00210$ & 0.22 & 0.58\\
Two-view EKF & $0.00970\!\pm\!0.00154$ & $0.00700\!\pm\!0.00091$ & $0.00262\!\pm\!0.00062$ & 3.29 & 1.00\\
Naive EKF & $0.07404\!\pm\!0.00206$ & $0.13992\!\pm\!0.00168$ & $0.04389\!\pm\!0.00085$ & 3.12 & 0.00\\
Single target packet & $0.03674\!\pm\!0.00474$ & $0.00253\!\pm\!0.00034$ & $0.00094\!\pm\!0.00016$ & 1.50 & 0.84\\
\hline
\end{tabular}
}\\[3pt]
\parbox{\columnwidth}{\raggedright\scriptsize Note: 50 trials per row, Gaussian noise, no outliers; $\pm$ entries are 95\% confidence intervals.}
\end{table}

In Table~\ref{tab:methods}, the target, beacon, and yaw columns report RMSE in meters, meters, and radians, and the time column is total wall-clock runtime per trial, including all candidate solves for multistart; the error statistics are seed-deterministic, while timings vary with the host machine.
The comparison separates estimation quality from initialization: the two-view EKF is viable after the constructive seed, while the naive EKF fails because the beacon pose remains poorly calibrated.
Sliding-window GN is fast but less reliable because short windows contain fewer distinct views and less spread.
Repeated target packets matter: a single target packet raises target RMSE from $0.00549$ m to $0.03674$ m in the $50$-trial comparison of Table~\ref{tab:methods}; in the separate $100$-trial Monte Carlo summary, adding a second beacon reduces batch target RMSE from $0.00563\pm0.00066$ m to $0.00366\pm0.00035$ m (Table~\ref{tab:mc}).

\begin{table}[t]
\centering
\caption{Monte Carlo model comparison.}
\label{tab:mc}
\scriptsize
\setlength{\tabcolsep}{1.0pt}
\resizebox{\columnwidth}{!}{%
\begin{tabular}{@{}l c c c c c@{}}
\hline
Model & Beacons & Target (m) & Beacon (m) & Yaw (rad) & Success\\
\hline
Local-frame model & 1 & $0.00563\!\pm\!0.00066$ & $0.00219\!\pm\!0.00024$ & $0.00083\!\pm\!0.00013$ & 1.00\\
Local-frame model & 2 & $0.00366\!\pm\!0.00035$ & $0.00248\!\pm\!0.00019$ & $0.00087\!\pm\!0.00008$ & 1.00\\
Calibrated baseline & 1 & $0.01026\!\pm\!0.00109$ & $0.00861\!\pm\!0.00096$ & -- & 1.00\\
Calibrated baseline & 2 & $0.00714\!\pm\!0.00072$ & $0.00642\!\pm\!0.00060$ & -- & 1.00\\
Naive EKF & 1 & $0.07377\!\pm\!0.00137$ & $0.14185\!\pm\!0.00128$ & $0.04416\!\pm\!0.00069$ & 0.00\\
Naive EKF & 2 & $0.03856\!\pm\!0.00160$ & $0.12206\!\pm\!0.00102$ & $0.02802\!\pm\!0.00046$ & 0.00\\
\hline
\end{tabular}
}\\[3pt]
\parbox{\columnwidth}{\raggedright\scriptsize Note: 100 trials per row. The calibrated baseline receives the target bearing directly in the global frame, eliminates the relay pose by back-substitution, and estimates the target alone; its beacon column back-projects positions from the calibrated packets, and it estimates no yaw.}
\end{table}

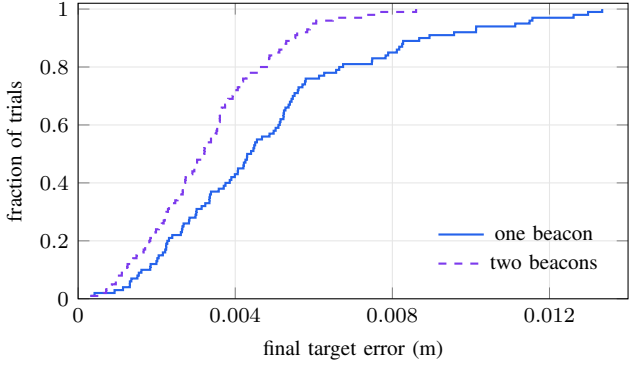
\begin{figure}[t]
\centering
\begin{tikzpicture}
\begin{axis}[
    width=\columnwidth, height=0.62\columnwidth,
    xlabel={final target error (m)}, ylabel={fraction of trials},
    xmin=0, xmax=0.014, ymin=0, ymax=1.02,
    scaled x ticks=false,
    xtick={0,0.004,0.008,0.012},
    xticklabel style={/pgf/number format/fixed, /pgf/number format/precision=3},
    grid=major, grid style={black!10},
    tick label style={font=\footnotesize},
    label style={font=\footnotesize},
    legend style={font=\footnotesize, at={(0.97,0.05)}, anchor=south east, draw=none, fill=none},
]
\addplot[const plot, thick, figblue] table[x=e1, y=p] {target_error_ecdf.dat};
\addlegendentry{one beacon}
\addplot[const plot, thick, figpurple, dashed] table[x=e2, y=p] {target_error_ecdf.dat};
\addlegendentry{two beacons}
\end{axis}
\end{tikzpicture}
\caption{Empirical CDF of per-trial final target error in the 100-trial Monte Carlo batch of Table~\ref{tab:mc}, for one and two beacons.}
\label{fig:ecdf}
\end{figure}

Table~\ref{tab:mc} compares the local-frame model against both baselines over 100 trials.
Two observations matter.
First, adding a second relay improves target accuracy by about a factor of $1.5$, consistent with Corollary~\ref{cor:augmentation}: the second beacon contributes an independent, self-calibrated view of the same target.
Fig.~\ref{fig:ecdf} shows the per-trial distributions behind the first two rows: the second beacon improves the whole distribution rather than the mean alone, and compresses the tail, with a worst trial of $0.0086$ m against $0.0133$ m for one beacon.
Second, the local-frame model outperforms the calibrated global-frame baseline even though the baseline is handed the relay calibration for free: after eliminating the relay pose, the baseline condenses each packet into a single range-consistency residual, while the local-frame estimator fuses all four whitened residual components per packet through the jointly estimated frame.
Self-calibration is therefore not a tax paid for realism; with enough excitation, the joint estimator recovers packet information that the calibrated reduction discards.
The naive-EKF rows repeat the initialization failure of Table~\ref{tab:methods} at 100 trials: without the constructive two-view seed, the filter linearizes about an uncalibrated relay pose and never recovers, even with a second beacon.

\subsection{Noise Robustness}

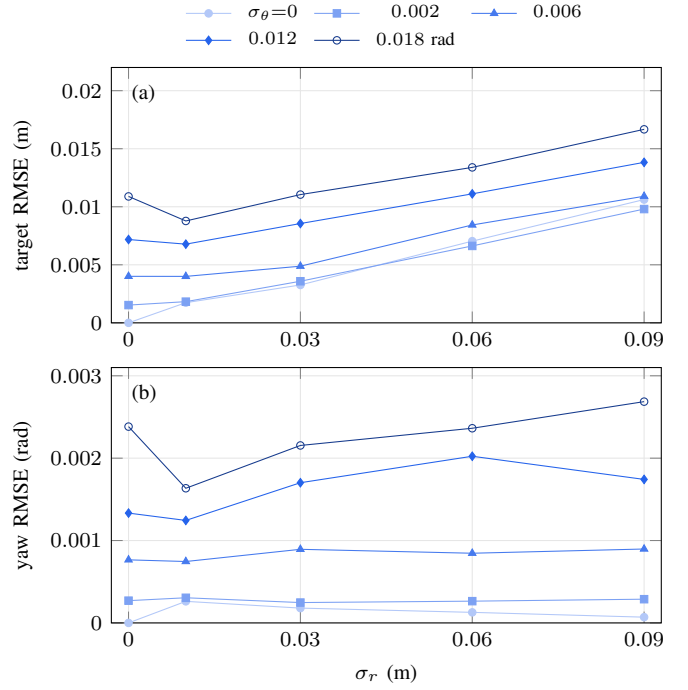
\begin{figure}[t]
\centering
\begin{tikzpicture}
\begin{axis}[
    width=\columnwidth, height=0.56\columnwidth,
    ylabel={target RMSE (m)},
    xmin=-0.003, xmax=0.093, ymin=0, ymax=0.022,
    scaled x ticks=false, scaled y ticks=false,
    xtick={0,0.03,0.06,0.09},
    xticklabel style={/pgf/number format/fixed, /pgf/number format/precision=2},
    ytick={0,0.005,0.01,0.015,0.02},
    yticklabel style={/pgf/number format/fixed, /pgf/number format/precision=3},
    grid=major, grid style={black!10},
    tick label style={font=\footnotesize},
    label style={font=\footnotesize},
    legend style={font=\scriptsize, at={(0.5,1.03)}, anchor=south, draw=none, fill=none, column sep=4pt},
    legend columns=3,
]
\addplot[figblue!35, mark=*, mark size=1.4pt] table[x=sr, y=t00] {noise_robustness_curves.dat};
\addlegendentry{$\sigma_\theta{=}0$}
\addplot[figblue!60, mark=square*, mark size=1.4pt] table[x=sr, y=t02] {noise_robustness_curves.dat};
\addlegendentry{$0.002$}
\addplot[figblue!85, mark=triangle*, mark size=1.6pt] table[x=sr, y=t06] {noise_robustness_curves.dat};
\addlegendentry{$0.006$}
\addplot[figblue, mark=diamond*, mark size=1.6pt] table[x=sr, y=t12] {noise_robustness_curves.dat};
\addlegendentry{$0.012$}
\addplot[figblue!60!black, mark=o, mark size=1.4pt] table[x=sr, y=t18] {noise_robustness_curves.dat};
\addlegendentry{$0.018$~rad}
\node[anchor=north west, font=\footnotesize] at (rel axis cs:0.02,0.97) {(a)};
\end{axis}
\end{tikzpicture}\\[2pt]
\begin{tikzpicture}
\begin{axis}[
    width=\columnwidth, height=0.56\columnwidth,
    xlabel={$\sigma_r$ (m)}, ylabel={yaw RMSE (rad)},
    xmin=-0.003, xmax=0.093, ymin=0, ymax=0.0031,
    scaled x ticks=false, scaled y ticks=false,
    xtick={0,0.03,0.06,0.09},
    xticklabel style={/pgf/number format/fixed, /pgf/number format/precision=2},
    ytick={0,0.001,0.002,0.003},
    yticklabel style={/pgf/number format/fixed, /pgf/number format/precision=3},
    grid=major, grid style={black!10},
    tick label style={font=\footnotesize},
    label style={font=\footnotesize},
]
\addplot[figblue!35, mark=*, mark size=1.4pt] table[x=sr, y=y00] {noise_robustness_curves.dat};
\addplot[figblue!60, mark=square*, mark size=1.4pt] table[x=sr, y=y02] {noise_robustness_curves.dat};
\addplot[figblue!85, mark=triangle*, mark size=1.6pt] table[x=sr, y=y06] {noise_robustness_curves.dat};
\addplot[figblue, mark=diamond*, mark size=1.6pt] table[x=sr, y=y12] {noise_robustness_curves.dat};
\addplot[figblue!60!black, mark=o, mark size=1.4pt] table[x=sr, y=y18] {noise_robustness_curves.dat};
\node[anchor=north west, font=\footnotesize] at (rel axis cs:0.02,0.97) {(b)};
\end{axis}
\end{tikzpicture}
\caption{One-beacon noise sweep: (a) target RMSE and (b) relay-yaw RMSE versus range noise $\sigma_r$, one curve per bearing noise $\sigma_\theta$, shaded light to dark with increasing $\sigma_\theta$ (50 trials per cell).}
\label{fig:noise}
\end{figure}

At $\sigma_\theta=0.006$ rad, the one-beacon noise sweep gives target RMSE $0.00843$ m for $\sigma_r=0.06$ m and $0.01091$ m for $\sigma_r=0.09$ m.
Fig.~\ref{fig:noise} reports the full grid behind these two values: target error degrades smoothly in both noise channels, the success rate is $1.00$ over the entire grid, and bearing noise is the dominant channel at the nominal $\sigma_r\le0.03$ m operating point.
Panel (b) isolates the self-calibration variable: the relay-yaw error is governed by $\sigma_\theta$, with only weak dependence on $\sigma_r$, matching the whitened residual structure in which yaw enters through the bearing rows.
There is no noise threshold behavior anywhere on the grid, consistent with Corollary~\ref{cor:gn}, which predicts first-order error growth once the excited trajectory has removed the gauge with margin.

\subsection{Initialization Robustness}

\begin{table}[t]
\centering
\caption{Initialization sweep.}
\label{tab:poorinit}
\scriptsize
\setlength{\tabcolsep}{3.0pt}
\begin{tabular}{l c c c c c}
\hline
Case & Offset (m) & Radius (m) & Yaw (rad) & Starts & Target (m)\\
\hline
Nominal seed & 0 & 2.0 & 0 & 1 & 0.00513\\
Target seed 1 m & 1 & 2.0 & 0 & 1 & 0.00498\\
Target seed 2 m & 2 & 2.0 & 0 & 1 & 0.00516\\
Poor beacon radius & 1 & 0.7 & 0 & 1 & 0.00493\\
Poor yaw seed & 1 & 2.0 & 1.57 & 1 & 0.00497\\
Poor seed, multistart & 2 & 2.8 & 2.4 & 6 & 0.00524\\
\hline
\end{tabular}\\[3pt]
\parbox{\columnwidth}{\raggedright\scriptsize Note: 50 trials per row; success is $1.00$ in every row. Offset is the target-seed offset from the true target; radius and yaw give the beacon seed circle and yaw guess; starts counts multistart candidates.}
\end{table}

A poor-initialization sweep with $2$ m target offset, $2.8$ m beacon seed radius, $2.4$ rad yaw seed, and six starts gives multistart GN target RMSE $0.00524$ m and success $1.00$.
Table~\ref{tab:poorinit} expands this sweep: every seeding corruption up to a $2$ m target offset, a wrong beacon-seed circle, and a $2.4$ rad yaw guess converges to the same $\approx5$ mm accuracy, and all but the last row do so from a single start.
This is the numerical counterpart of Corollary~\ref{cor:gn}: once excitation removes the gauge, the attraction basin around the isolated minimizer is wide enough that crude seeds converge, and multistart acts as a safeguard rather than a necessity.

\subsection{Outliers, Dropout, and Vehicle-Pose Error}

\begin{table}[t]
\centering
\caption{One-beacon robustness and failure modes.}
\label{tab:stress}
\scriptsize
\begin{tabular}{l c c c}
\hline
Case & Setting & Target (m) & Success\\
\hline
Stationary & rank 4, $S_v=0$ & 6.068 & 0.00\\
Short line & $\kappa=117.0$ & 0.025 & 0.82\\
Repeated views & $\kappa=130.3$ & 0.048 & 0.70\\
60\% dropout & 30.7 packets & 0.008 & 1.00\\
10\% outliers, GN & $0.75$ m, $0.35$ rad & 0.097 & 0.10\\
10\% outliers, Huber & same corruptions & 0.0065 & 1.00\\
i.i.d.\ $q_k$ noise & $\sigma_q=0.10$ m & 0.064 & 0.44\\
random-walk drift & $0.5\%$ step drift & 0.0096 & 1.00\\
\hline
\end{tabular}
\end{table}

Table~\ref{tab:stress} locates the practical boundary: weak spread is ill conditioned, diverse packets tolerate dropout, Huber weighting rejects heavy-tailed packets, and vehicle-pose error enters beacon and target recovery directly.
The outlier rows corrupt 10\% of packets by $0.75$ m and $0.35$ rad; the drift row uses $e_k=e_{k-1}+\nu_k$, $\nu_k\sim\mathcal N(0,\sigma_k^2I)$, with $\sigma_k=0.005\|q_k-q_{k-1}\|$.
Drifting odometry, asynchronous packets, and moving targets require additional states and gauge-consistent estimation.

\begin{table}[t]
\centering
\caption{Outlier robustness sweep.}
\label{tab:outliers}
\scriptsize
\setlength{\tabcolsep}{4.0pt}
\begin{tabular}{c c c c c}
\hline
& \multicolumn{2}{c}{Batch GN} & \multicolumn{2}{c}{Huber GN}\\
Outliers & Target (m) & Success & Target (m) & Success\\
\hline
0\% & 0.00626 & 1.00 & 0.00633 & 1.00\\
5\% & 0.06576 & 0.46 & 0.00513 & 1.00\\
10\% & 0.09676 & 0.10 & 0.00654 & 1.00\\
20\% & 0.13107 & 0.08 & 0.00768 & 1.00\\
\hline
\end{tabular}\\[3pt]
\parbox{\columnwidth}{\raggedright\scriptsize Note: 50 trials per row; corrupted packets are shifted by $0.75$ m in range and $0.35$ rad in bearing.}
\end{table}

\begin{table}[t]
\centering
\caption{Dropout and vehicle-pose error sweeps.}
\label{tab:degradation}
\scriptsize
\setlength{\tabcolsep}{4.0pt}
\begin{tabular}{l c c c}
\hline
Case & Target (m) & Beacon (m) & Success\\
\hline
Dropout 0\% (80.0 pkts) & 0.00541 & 0.00231 & 1.00\\
Dropout 15\% (67.8 pkts) & 0.00646 & 0.00254 & 1.00\\
Dropout 30\% (55.0 pkts) & 0.00637 & 0.00291 & 1.00\\
Dropout 45\% (44.4 pkts) & 0.00678 & 0.00310 & 1.00\\
Dropout 60\% (30.7 pkts) & 0.00815 & 0.00344 & 1.00\\
\hline
$\sigma_q=0$ & 0.00598 & 0.00223 & 1.00\\
$\sigma_q=0.10$ m & 0.06434 & 0.05755 & 0.44\\
$\sigma_q=0.20$ m & 0.16835 & 0.16213 & 0.02\\
$\sigma_q=0.30$ m & 0.30748 & 0.28355 & 0.00\\
$\sigma_q=0.40$ m & 0.47450 & 0.43509 & 0.00\\
Drift $0.5\%$ & 0.00963 & 0.00841 & 1.00\\
\hline
\end{tabular}\\[3pt]
\parbox{\columnwidth}{\raggedright\scriptsize Note: 50 trials per row; each block is a separate sweep with its own zero-degradation reference row. Dropout rows give the mean delivered packets out of 80; $\sigma_q$ rows add i.i.d.\ noise to the stored vehicle poses; the drift row uses the random-walk model above.}
\end{table}

Tables~\ref{tab:outliers} and~\ref{tab:degradation} expand the corresponding rows of Table~\ref{tab:stress}.
The outlier sweep shows the unprotected failure is graded: batch GN success decays from $0.46$ to $0.08$ as the corruption rate grows, while Huber weighting holds millimeter accuracy and full success through $20\%$ corruption, at essentially no cost on clean data (first row).
The degradation sweeps separate the two informational assumptions of model \eqref{eq:model}.
Dropout thins the window but leaves the anchor poses exact, and accuracy degrades gently, from $0.0054$ m to $0.0082$ m at $60\%$ dropout, because the surviving packets still span a wide spread.
Vehicle-pose noise instead corrupts the anchors themselves and is the binding assumption: $\sigma_q=0.10$ m already halves the success rate, and the beacon and target errors grow together because anchor error propagates through the recovered frame, as the sensitivity expression of Corollary~\ref{cor:gn} predicts.
The random-walk drift row stays benign because its accumulated error remains at the few-millimeter level over the window.

\subsection{Interpretation Across Stress Regimes}

The experiments separate three questions that a rank test alone cannot answer.
First, structural recoverability asks whether the gauge is removed: the stationary row fails exactly, while every nonstationary row is rank five.
Second, numerical recoverability asks whether the remaining information is large enough relative to noise: the short-line and repeated-view trajectories are identifiable but unreliable, and their small $S_v$ values, large condition numbers, and reduced success rates move together as Lemma~\ref{lem:spread} and Proposition~\ref{prop:polarinfo} predict.
Third, statistical robustness asks whether the assumed packets and anchor poses are trustworthy after geometry is adequate.
Dropout mainly reduces sample count while preserving a wide baseline, whereas outliers violate the Gaussian packet model and vehicle-pose error perturbs the global anchors that define the recovered frame.

This separation also clarifies what additional sensors buy.
Repeated target packets reduce the averaging component of target uncertainty but add no calibration rank; a second relay contributes an independent calibrated view after the anchor relay fixes the target; and neither mechanism repairs a weak anchor trajectory by itself.
The full results therefore support a layered design rule: first enforce $S_v>0$ to remove the gauge, then choose a spread margin using the native-noise information bound, and finally select robust losses or joint vehicle-pose states according to the anticipated non-Gaussian and navigation errors.

\subsection{Two-Beacon Augmentation Geometry}

\begin{table}[t]
\centering
\caption{Two-beacon separation sweep.}
\label{tab:geometry}
\scriptsize
\setlength{\tabcolsep}{4.0pt}
\begin{tabular}{c c c c}
\hline
Separation (m) & Target (m) & Beacon (m) & Yaw (rad)\\
\hline
0.3 & 0.00311 & 0.00214 & 0.00073\\
0.6 & 0.00298 & 0.00211 & 0.00073\\
1.0 & 0.00296 & 0.00202 & 0.00075\\
1.6 & 0.00283 & 0.00215 & 0.00074\\
2.4 & 0.00290 & 0.00195 & 0.00079\\
3.2 & 0.00295 & 0.00177 & 0.00073\\
4.0 & 0.00349 & 0.00244 & 0.00086\\
\hline
\end{tabular}\\[3pt]
\parbox{\columnwidth}{\raggedright\scriptsize Note: 50 trials per row on the excited trajectory; success is $1.00$ in every row.}
\end{table}

Corollary~\ref{cor:augmentation} says a second beacon needs only one noncoincident packet to join the calibrated solution, and the auxiliary-beacon margin remark places its sensitivity in the local target/vehicle separation, not in the inter-beacon distance.
Table~\ref{tab:geometry} confirms this: sweeping the two-beacon separation across more than a decade leaves target RMSE within $0.0028$--$0.0035$ m with full success everywhere.
There is no preferred baseline between the relays because each relay is calibrated by the same vehicle trajectory, not by the other relay.

\section{Degenerate Configurations and Equivariance}
\label{sec:scope}

The constructive result is global within its stated planar model, but its assumptions identify several distinct failure mechanisms that are useful when applying it.
First, one stored pose, or any number of repeated identical poses, gives $S_v=0$ and retains the exact one-parameter gauge of Theorem~\ref{thm:gauge}.
Second, $S_v>0$ removes that gauge, but arbitrarily small spread remains arbitrarily ill conditioned by Lemma~\ref{lem:spread} and Proposition~\ref{prop:polarinfo}; identifiability is therefore not an accuracy guarantee.
Third, a zero vehicle--beacon or target--beacon range makes the corresponding bearing undefined and lies outside the native polar model.
Fourth, vehicle packets alone can calibrate $(x,\psi)$ after two views but cannot recover $p$ without at least one associated target packet; repeated measurements of the same static target vector improve averaging but do not add rank.
Fifth, yaw is unique only modulo $2\pi$ because bearings are wrapped; a reflected solution is excluded because the relay transform is restricted to $SO(2)$ rather than $O(2)$.
For several relays, one excited anchor fixes the shared target, after which Corollary~\ref{cor:augmentation}'s noncoincident packet condition must hold separately for every additional relay.

The model is exactly equivariant to a common global translation.
For any $c\in\mathbb R^2$, replace $(q_k,x,p)$ by $(q_k+c,x+c,p+c)$ while leaving $\psi$ unchanged.
Every local vector in \eqref{eq:model}, every native range-bearing prediction, $S_v$, and the weighted residual cost remain identical.
Accordingly, the estimator errors and Fisher information are unchanged under a change of global origin; this algebraic invariance should not be confused with moving only the vehicle trajectory relative to fixed target and relay positions, which changes the ranges and therefore the native angular-noise weighting in Proposition~\ref{prop:polarinfo}.

\section{Statistical Protocol and Reproducibility}
\label{sec:reproducibility}

All reported estimators use the same native range-bearing residual implementation and analytic Jacobian.
The constructive initializer searches the stored window for the pair with largest local-vector separation, computes $(\hat\psi,\hat x,\hat p)$ by Theorem~\ref{thm:constructive}, and only then permits nonlinear refinement.
The nominal Monte Carlo study uses deterministic seeds, 80 known vehicle poses, 100 trials per primary row, and 50 trials per expanded stress row; confidence intervals are computed from trial-level errors, while success is evaluated independently using the fixed $0.05$-m target, $0.05$-m relay-position, and $0.05$-rad yaw criteria.
Outlier comparisons share noise seeds and corrupted packet indices across vanilla and Huber estimators, and paired model comparisons share the underlying world and trajectory.
This matched design makes differences between estimators paired outcomes rather than comparisons between unrelated sample populations.
Geometric quantities ($S_v$, rank, $\sigma_{\min}$, and $\kappa$) are evaluated noise-free at the true state, while accuracy and success are computed from noisy trial-level estimates; the separation prevents random solver behavior from being mistaken for a change in observability.
Solver termination is recorded independently of the accuracy thresholds, so numerical stopping is never counted as successful localization by itself.
The accompanying artifact contains the estimator implementation, deterministic experiment definitions, trial-level outputs, and figure-generation workflow needed to reproduce the reported statistics.

\section{Conclusion}

This paper established that vehicle motion can self-calibrate an unknown-pose range-bearing relay and make its hidden-target packet globally actionable.
The key point is a minimal-excitation characterization, not only a rank condition: two distinct local vehicle observations make one unknown-pose relay sufficient, via a constructive formula for its yaw, position, and the anchored target.
The centered spread $S_v$ links identifiability to estimator conditioning and excitation design, while simulations confirm exact noiseless recovery, accurate noisy recovery with confidence intervals, two-view EKF performance, and predictable degradation under weak motion, outliers, single target packets, and vehicle-pose error.
The anchored excitation-supervision and anchor-free correlated-odometry extensions in~\cite{bagla2027excitationsupervised,bagla2026anchorfree} show how the same calibration geometry supports distinct closed-loop operating regimes.
Remaining directions include asynchronous packets, three-dimensional relay geometry, and physical robot validation.

\bibliographystyle{IEEEtran}
\bibliography{references}

\begin{thebibliography}{10}
\providecommand{\url}[1]{#1}
\csname url@samestyle\endcsname
\providecommand{\newblock}{\relax}
\providecommand{\bibinfo}[2]{#2}
\providecommand{\BIBentrySTDinterwordspacing}{\spaceskip=0pt\relax}
\providecommand{\BIBentryALTinterwordstretchfactor}{4}
\providecommand{\BIBentryALTinterwordspacing}{\spaceskip=\fontdimen2\font plus
\BIBentryALTinterwordstretchfactor\fontdimen3\font minus
  \fontdimen4\font\relax}
\providecommand{\BIBforeignlanguage}[2]{{%
\expandafter\ifx\csname l@#1\endcsname\relax
\typeout{** WARNING: IEEEtran.bst: No hyphenation pattern has been}%
\typeout{** loaded for the language `#1'. Using the pattern for}%
\typeout{** the default language instead.}%
\else
\language=\csname l@#1\endcsname
\fi
#2}}
\providecommand{\BIBdecl}{\relax}
\BIBdecl

\bibitem{guler2017adaptive}
S.~G{\"u}ler, B.~Fidan, S.~Dasgupta, B.~D.~O. Anderson, and I.~Shames,
  ``Adaptive source localization based station keeping of autonomous
  vehicles,'' \emph{IEEE Trans. Autom. Control}, vol.~62, no.~7, pp.
  3122--3135, 2017.

\bibitem{roumeliotis2002distributed}
S.~I. Roumeliotis and G.~A. Bekey, ``Distributed multirobot localization,''
  \emph{IEEE Trans. Robot. Autom.}, vol.~18, no.~5, pp. 781--795, 2002.

\bibitem{durrantwhyte2006slam}
H.~Durrant-Whyte and T.~Bailey, ``Simultaneous localization and mapping: {Part
  I},'' \emph{IEEE Robot. Autom. Mag.}, vol.~13, no.~2, pp. 99--110, 2006.

\bibitem{cadena2016slam}
C.~Cadena, L.~Carlone, H.~Carrillo, Y.~Latif, D.~Scaramuzza, J.~Neira, I.~Reid,
  and J.~J. Leonard, ``Past, present, and future of simultaneous localization
  and mapping: Toward the robust-perception age,'' \emph{IEEE Trans. Robot.},
  vol.~32, no.~6, pp. 1309--1332, 2016.

\bibitem{huang2010observability}
G.~P. Huang, A.~I. Mourikis, and S.~I. Roumeliotis, ``Observability-based rules
  for designing consistent {EKF} {SLAM} estimators,'' \emph{Int. J. Robot.
  Res.}, vol.~29, no.~5, pp. 502--528, 2010.

\bibitem{barooah2007estimation}
P.~Barooah and J.~P. Hespanha, ``Estimation on graphs from relative
  measurements,'' \emph{IEEE Control Syst. Mag.}, vol.~27, no.~4, pp. 57--74,
  2007.

\bibitem{eren2004rigidity}
T.~Eren, D.~K. Goldenberg, W.~Whiteley, Y.~R. Yang, A.~S. Morse, B.~D.~O.
  Anderson, and P.~N. Belhumeur, ``Rigidity, computation, and randomization in
  network localization,'' in \emph{IEEE INFOCOM}, vol.~4, 2004, pp. 2673--2684.

\bibitem{arun1987least}
K.~S. Arun, T.~S. Huang, and S.~D. Blostein, ``Least-squares fitting of two
  {3-D} point sets,'' \emph{IEEE Trans. Pattern Anal. Mach. Intell.}, vol.
  PAMI-9, no.~5, pp. 698--700, 1987.

\bibitem{horn1987closed}
B.~K.~P. Horn, ``Closed-form solution of absolute orientation using unit
  quaternions,'' \emph{J. Opt. Soc. Am. A}, vol.~4, no.~4, pp. 629--642, 1987.

\bibitem{shames2013bearing}
I.~Shames, A.~N. Bishop, and B.~D.~O. Anderson, ``Analysis of noisy
  bearing-only network localization,'' \emph{IEEE Trans. Autom. Control},
  vol.~58, no.~1, pp. 247--252, 2013.

\bibitem{piovan2013frame}
G.~Piovan, I.~Shames, B.~Fidan, F.~Bullo, and B.~D.~O. Anderson, ``On frame and
  orientation localization for relative sensing networks,'' \emph{Automatica},
  vol.~49, no.~1, pp. 206--213, 2013.

\bibitem{dou2020target}
L.~Dou, C.~Song, X.~Wang, L.~Liu, and G.~Feng, ``Target localization and
  enclosing control for networked mobile agents with bearing measurements,''
  \emph{Automatica}, vol. 118, p. 109022, 2020.

\bibitem{han2016mobileanchor}
G.~Han, J.~Jiang, C.~Zhang, T.~Q. Duong, M.~Guizani, and G.~K. Karagiannidis,
  ``A survey on mobile anchor node assisted localization in wireless sensor
  networks,'' \emph{IEEE Commun. Surveys Tuts.}, vol.~18, no.~3, pp.
  2220--2243, 2016.

\bibitem{jia2025target}
T.~Jia, X.~Ke, H.~Liu, K.~C. Ho, and H.~Su, ``Target localization and sensor
  self-calibration of position and synchronization by range and angle
  measurements,'' \emph{IEEE Trans. Signal Process.}, vol.~73, pp. 340--355,
  2025.

\bibitem{tsai1989new}
R.~Y. Tsai and R.~K. Lenz, ``A new technique for fully autonomous and efficient
  {3D} robotics hand/eye calibration,'' \emph{IEEE Trans. Robot. Autom.},
  vol.~5, no.~3, pp. 345--358, 1989.

\bibitem{bagla2019receding}
Y.~Bagla and V.~Srivastava, ``On receding horizon chance constraint motion
  planning for uncertain multi-agent systems,'' in \emph{Proc. ASME Dyn. Syst.
  Control Conf.}, vol. 59162, 2019, p. V003T19A012.

\bibitem{johnson2024impact}
C.~Johnson, Y.~Bagla, K.~Borle, A.~Katpatal, M.~Mahmood, and E.~B. Olson,
  ``Method and system for impact-based operation of an autonomous agent,'' U.S.
  Patent 12\,012\,123, Jun., 2024.

\bibitem{leny2018localizability}
J.~Le~Ny and S.~Chauvi{\`e}re, ``Localizability-constrained deployment of
  mobile robotic networks with noisy range measurements,'' \emph{IEEE Trans.
  Robot.}, vol.~34, no.~3, pp. 705--721, 2018.

\bibitem{nocedal2006numerical}
J.~Nocedal and S.~J. Wright, \emph{Numerical Optimization}, 2nd~ed.\hskip 1em
  plus 0.5em minus 0.4em\relax Springer, 2006.

\bibitem{barshalom2001estimation}
Y.~Bar-Shalom, X.~R. Li, and T.~Kirubarajan, \emph{Estimation with Applications
  to Tracking and Navigation}.\hskip 1em plus 0.5em minus 0.4em\relax Wiley,
  2001.

\bibitem{anderson2010formal}
B.~D.~O. Anderson, I.~Shames, G.~Mao, and B.~Fidan, ``Formal theory of noisy
  sensor network localization,'' \emph{SIAM J. Discrete Math.}, vol.~24, no.~2,
  pp. 684--698, 2010.

\bibitem{khalil2002nonlinear}
H.~K. Khalil, \emph{Nonlinear Systems}, 3rd~ed.\hskip 1em plus 0.5em minus
  0.4em\relax Prentice Hall, 2002.

\bibitem{bagla2027excitationsupervised}
Y.~Bagla, ``Excitation-supervised closed-loop self-calibration and target
  seeking for an unknown-pose range-bearing relay,'' 2026, arXiv:2608.12528.

\bibitem{bagla2026anchorfree}
------, ``Anchor-free hidden-target seeking via certified self-calibration
  under correlated odometry,'' 2026, arXiv:2609.27905.

\end{thebibliography}

\end{document}